\documentclass[english,twoside,11pt]{article}

\usepackage[lmargin=1in,rmargin=1in,tmargin=1in,bmargin=1in]{geometry}

\usepackage{float}
\usepackage{graphicx}
\usepackage{color}
\usepackage{tikz}
\usetikzlibrary{backgrounds}

\usepackage[margin=20pt]{caption}

\usepackage{amsmath}
\usepackage{amssymb}
\usepackage{amsthm}
\usepackage{mathrsfs}

\usepackage[english]{babel}
\usepackage{flafter}
\usepackage{array}
\usepackage{paralist}

\graphicspath{{.}{graphics/}}

\newtheorem{theorem}{Theorem}
\newtheorem{corollary}[theorem]{Corollary}
\newtheorem*{corollary*}{Corollary}
\newtheorem{lemma}[theorem]{Lemma}
\newtheorem{proposition}[theorem]{Proposition}
\newtheorem*{proposition*}{Proposition}

\newtheorem*{property*}{Property}

\theoremstyle{definition}
\newtheorem{definition}[theorem]{Definition}
\newtheorem*{definition*}{Definition}

\theoremstyle{remark}
\newtheorem{remark}[theorem]{Remark}
\newtheorem*{remark*}{Remark}

\newtheorem*{example*}{Example}

\usepackage{framed}
\usepackage{algorithm}
\usepackage[noend]{algorithmic}

\input{mathdots}

\title{Fault-Tolerant Shortest Paths - Beyond the Uniform Failure Model}

\author{
\begin{minipage}{\linewidth}
\begin{center}
David Adjiashvili\\[0.3\baselineskip]
\normalsize Institute for Operations Research, ETH Zurich\\
\normalsize david.adjiashvili@ifor.math.ethz.ch
\end{center}
\end{minipage}
}

\date{\today}

\begin{document}

%
%

\maketitle

\begin{abstract}
The overwhelming majority of survivable (fault-tolerant) network design models assume a uniform scenario set. 
Such a scenario set assumes that every subset of the network resources (edges or vertices) of a given 
cardinality $k$ comprises a scenario. While this approach yields problems with clean combinatorial structure 
and good algorithms, it often fails to capture the true nature of the scenario set coming from applications. 

One natural refinement of the uniform model is obtained by partitioning the set of resources into \textit{faulty} 
and {secure} resources. The scenario set contains every subset of at most $k$ faulty resources. This work 
studies the \textit{Fault-Tolerant Path} (FTP) problem, the counterpart of the Shortest Path problem in 
this failure model. We present complexity results alongside exact and approximation algorithms for FTP. 
We emphasize the vast increase in the complexity of the problem with respect to its uniform analogue, 
the Edge-Disjoint Paths problem.

\end{abstract}

\section{Introduction}\label{sec:intro}
The \textit{Minimum-Cost Edge-Disjoint Path} (EDP) problem is a classical network design problem, defined as follows.
Given an edge-weighted graph $G=(V,E)$, two terminals $s,t\in V$ and an integer parameter $k\in \mathbb{Z}_+$, find
$k$ edge-disjoint paths connecting $s$ and $t$ with minimum total cost. EDP is motivated by the following survivable network
design problem: what is the connection cost of two nodes in a network, given that any $k-1$ edges can be a-posteriori removed
from the graph. The implicit assumption in EDP is that every edge in the graph in equally vulnerable. Unfortunately, this
assumption is unrealistic in many applications, hence resulting in overly-conservative solutions. 
This paper studies a natural refinement of the EDP problem called the \textit{Fault-Tolerant Path Problem} (FTP), formally
defined as follows.

\begin{framed}
\textbf{Instance: } An edge-weighted graph $G=(V,E)$, two nodes $s,t\in V$, a subset of the edges $M\subset E$ and an integer
$k\in \mathbb{Z}_+$.

\textbf{Problem: } Find minimum cost $S \subset E$ such that $S \setminus F$ contains an $s$-$t$ path for every $F\subset M$ with
$|F| \leq k$.
\end{framed}

Note that FTP becomes EDP when $M = E$. We let $\Omega(M,k)$ denote the set of possible faults, namely $\Omega(M,k) = \{A\subset M: |A| \leq k\}$.

A well-known polynomial algorithm for EDP works as follows. Assign unit capacities to all edges in $G$ and find a minimum-cost
$k$-flow from $s$ to $t$. The integrality property of the \textit{Minimum-Cost Flow} (MCF) problem guarantees that an
extreme-point optimal solution is integral, hence it corresponds to a subset of edges. It is then straightforward to verify that
this set is an optimal solution of the EDP problem (for a thorough treatment of this method we refer to the book of Schrijver~\cite{Schrijver}).

The latter algorithm raises two immediate questions concerning FTP. The first question is whether FTP admits a polynomial
time algorithm. In this paper we give a negative answer to this question conditioned on P$\neq$NP, showing that FTP is NP-hard. 
In fact, the existence of constant-factor approximation algorithms is unlikely already for the restricted case of directed 
acyclic graphs. Consequently, it is natural to ask whether polynomial algorithms can be obtained for restricted variants of FTP. 
We devote the first part of this paper to this question.

The second question
concerns the natural fractional relaxation of FTP. As we previously observed, one natural relaxation of EDP is the
MCF problem. This relaxation admits an integrality gap of one, namely the optimal integral solution value is always equal to 
the corresponding optimal fractional value. In the context of FTP it is natural to study the following fractional relaxation,
which we denote by \textit{Fractional FTP} (FRAC-FTP).

\begin{framed}
\textbf{Instance: } An edge weighted graph $G=(V,E)$, two nodes $s,t\in V$, a subset of the edges $M\subset E$ and an integer
$k\in \mathbb{Z}_+$.

\textbf{Problem: } Find a minimum cost capacity vector $x: E \rightarrow [0,1]$ such that for every $F\in \Omega(M,k)$, the
maximum $s$-$t$ flow in $G_F = (V, E\setminus F)$, capacitated by $x$, is at least one. 
\end{framed}

It is obvious that by further requiring $x \in \{0,1\}^E$ we obtain the FTP problem, thus FRAC-FTP is indeed a fractional
relaxation of FTP. We later show that the integrality gap with respect to this relaxation is bounded by $k+1$. Furthermore, we show that
this bound is tight, namely that there exists an infinite family of instances with integrality gap arbitrarily close to $k+1$.
This result also leads to a simple $(k+1)$-approximation algorithm for FTP, which we later combine with an
algorithm for the case $k=1$ to obtain a $k$-approximation algorithm.

The remainder of the paper is organized as follows. In the following section we review some related work.
In Section~\ref{sec:complexity} we study the complexity of FTP.
In Section~\ref{sec:exact} we provide some exact polynomial algorithms for special cases of FTP. In
Section~\ref{sec:intgap} we relate FTP and FRAC-FTP by proving a tight bound on the interagrality gap and
show how this result leads to a $k$-approximation algorithm for FTP. In Section~\ref{sec:conclusions} we
conclude the paper and mention some interesting directions for future work.

\section{Related Work}\label{sec:related}

The shortest path problem was studied in numerous robust settings 
(\cite{RSP,CompRSP,DemandR,DemandR2,MinMaxRegretSP,RRSP,ExactSubgraphSP}). Yu and Yang~\cite{RSP}
studied the shortest path problem with a discrete scenario set in the minimax and minimax-regret models.
They show that both problems are NP-Hard even for $2$ scenarios and layered graphs.
Aissi et al.~\cite{MinMaxRegretSP} extended these results by showing that both version of the
problem admit an FPTAS if the number of scenarios is bounded. They also show that both problems are 
not approximable within $2-\epsilon$ for any $\epsilon > 0$ with an unbounded number of scenarios unless P$=$NP.
Adjiashvili et al.~\cite{Adjiashvili} studied a two-stage feasibility problem
in which for a given pair of terminals $s$ and $t$ in a graph $G=(V,E)$, and two parameters $f,r\in \mathbb{Z}_+$, the goal
is to choose a minimum-cardinality subset of the edges $S\subset E$ with the property that for every failure set
$F\subset E$ with $|F| \leq f$ there exists a recovery set $R\subset E \setminus F$ with $|R|\leq r$,
such that $s$ and $t$ are connected in $(S \setminus F) \cup R$. The authors give an exact algorithm for
$f=r=1$ and a $2$-approximation algorithm for $f=1$ and arbitrary $r$. Also some complexity results
are provided.
Puhl~\cite{RRSP} studied the recoverable-robust shortest path problem, a robust path problem in a
two-stage optimization setting. Hardness results were provided for discrete scenario sets, interval
scenarios and $\Gamma$-scenarios for two variants of cost function. For one variant an approximation 
algorithm is given.
B\"using~\cite{ExactSubgraphSP}
developed approximation algorithms and proved hardness results for the problem of finding the smallest 
number of edges in the graph which contains a shortest path according to every scenario.

The robustness model proposed in this paper is natural for various classical combinatorial optimization
problems. Of particular interest is the counterpart of the Minimum Spanning Tree problem. This problem
is closely related to the Minimum $k$-Edge Connected Spanning Subgraph (ECSS) problem, a well-understood
robust connection problem.
Gabow, Goemans, Tardos and Williamson~\cite{kEdgeConnected1}
developed a polynomial time $(1 + \frac{c}{k})$-approximation algorithm for ECSS, for some fixed
constant $c$. The authors also show that for some constant $c' < c$, the existence of a polynomial time
$(1 + \frac{c'}{k})$-approximation algorithm implies P$=$NP. An intriguing property of ECSS is
that the problem becomes easier to approximate when $k$ grows. Concretely, while for every fixed $k$,
ECSS is NP-hard to approximate within some factor $\alpha_k>1$, the latter result asserts that there 
is function $\beta(k)$ tending
to one as $k$ tends to infinity such that ECSS is approximable within a factor $\beta(k)$.
This phenomenon was already discovered by
Cheriyan and Thurimella~\cite{CheriyanThurimella}, who gave algorithms with a weaker approximation guarantee.
The more general Generalized Steiner Network problem admits a polynomial $2$-approximation algorithm due to
Jain~\cite{Jain}.

We review next some of the related research in robust combinatorial optimization. 
Liebchen et al.~\cite{RRFirst} proposed the notion of \emph{recoverable robustness} as a tractable
framework for robustness in railway optimization. The authors present
a general recipe for obtaining two-stage robust problems from any combinatorial optimization
problem. In~\cite{AdjRC} Ben-Tal et al.~proposed a method to delay some decisions in an
uncertain convex optimization problem to the stage at which the real data is revealed. Their
approach does not force some decision variables $X$ to be fixed as first-stage
decision variables. Instead, these variables are allowed to be associated in the first stage decision
with affine functions $X(\zeta)$ of the uncertain input vector $\zeta$. The resulting formulation
turns out to be theoretically and practically tractable in many cases. The framework of Bertsimas and
Caramanis~\cite{FiniteadAptability} allows to deal with discrete changes in the
second stage decision in the special case of linear optimization.

Kouvalis and Yu~\cite{RDOBook} provide an extensive study of complexity and algorithms for robust 
counterparts of many classical discrete optimization. Instead of resource removal, these models 
assume uncertainty in final cost of the resources.
Bertsimas and Sim~\cite{RDO} proposed a framework for incorporating robustness in discrete optimization. 
They show that robust counterparts of mixed integer linear programs with uncertainty in the matrix coefficients 
can be solved using moderately larger mixed integer linear programs with no uncertainty. 

A two-stage robust model called\textit{demand-robust optimization} was introduced by
Dhamdhere et al.~\cite{DemandR}. In this model a scenario corresponds to a subset
of the initially given constraints that need to be satisfied. For shortest
path this means that the source node and the target node are not fixed, but are rather
revealed in the second stage. The authors provide approximation 
algorithms for some problems such as Steiner Tree, Multi-Cut, facility Location etc.~in the
demand-robust setting. Golovin et al.~\cite{DemandR2} later improved some of those results,
including a constant factor approximation algorithm for the robust
shortest path problem. Feige et al.~\cite{ExpSce1} and Khandekar et al.~\cite{ExpSce2} studied
an extension of the demand-robust model that allows exponentially many scenarios.

\section{Complexity of FTP}\label{sec:complexity}

Our first observation is that FTP generalizes the \textit{Directed $m$-Steiner Tree Problem} ($m$-DST). The input
to $m$-DST is a weighted directed graph $G = (V,E)$, a \textit{source} node $s\in V$, a collections of \textit{terminals}
$T\subset V$ and an integer $m \leq |T|$. The goal is to find a minimum-cost arboresence $A \subset E$ rooted at $s$, that contains a 
directed path from $s$ to some subset of $m$ terminal. 

The $m$-DST is seen to be a special case of FTP as follows. Given an instance $I = (G, s,T,m)$ of $m$-DST 
define the following instance of FTP. The graph $G$ is augmented by $|T|$ new zero-cost directed edges $E'$ 
connecting every terminal $u\in T$ to a new node $t$. Finally, we set $M = E'$ and $k = m-1$. 
The goal is to find a fault-tolerant path from $s$ to $t$ in the new graph. It is now straightforward to see
that a solution $S$ to the FTP instance is feasible if and only if $S\cap E$ contains a feasible solution
to the $k$-DST problem (we can assume that all edges in $E'$ are in any solution to the FTP instance).

The latter observation implies an immediate conditional lower bound on the approximability of FTP. Halperin
and Krauthgamer~\cite{HalperinKrauthgamer} showed that $m$-DST cannot be approximated within a factor
$\log ^{2-\epsilon} m$ for every $\epsilon > 0$, unless NP $\subset$ ZTIME($n^{polylog(n)}$). As a result
we obtain the following.

\begin{proposition}\label{prop:complexity}
For every $\epsilon > 0$ FTP admits no polynomial approximation algorithms with ratio $\log ^{2-\epsilon} k$, 
unless NP $\subset$ ZTIME($n^{polylog(n)}$).
\end{proposition}

\begin{remark}\label{rem:alg:mneart} 
The reduction above can be easily adapted to obtain a $k^\epsilon$-approximation algorithm for 
FTP for the special case that $M \subset \{e\in E: t\in e\}$ using the algorithm of Charikar et. al.~\cite{CharikarChekuriETC}.
\end{remark}

In fact, any approximation algorithm with factor $\rho(k)$ for FTP is an approximation algorithm with factor $\rho(m)$ 
for $m$-DST. The best known algorithm for the latter problem is due to Charikar et. al.~\cite{CharikarChekuriETC}.
Their result is an approximation scheme attaining the approximation factor of $m^\epsilon$ for every $\epsilon > 0$.

In light of the previous discussion, it seems that obtaining an approximation algorithm for FTP
with a poly-logarithmic approximation guarantee is a challenging task. In fact, although we have
no concrete result to support this possibility, such an algorithm might not exist. 
We end this discussion with another 
special case of FTP, which we denote by \textit{Simultaneous Directed $m$-Steiner Tree} ($m$-SDST). An 
input to $m$-SDST specifies two edge-weighted graphs $G_1 = (V,E_1,w_1)$ and $G_2=(V,E_2,w_2)$ on the same
set of vertices $V$, a source $s$, a set of terminals $T\subset V$ and an integer $m \leq |T|$. The goal
is to find a subset $U\subset T$ of $m$ terminals and two arboresences $S_1\subset E_1$ and $S_2\subset E_2$ connecting
$s$ to $U$ in the respective graphs, so as to minimize $w_1(S_1) + w_2(S_2)$. $m$-SDST is seen to be a special case
of FTP via the following reduction. Given an instance of $m$-SDST, construct a graph $G = (V',E)$ as follows. 
Take a disjoint union of $G_1$ and $G_2$, where 
the direction of every edge in $G_2$ is reversed. Connect every copy of a terminal $u\in T$ in $G_1$ to its 
corresponding copy in $G_2$ with an additional zero-cost edge $e_u$. Finally, set $M = \{e_u : u\in T\}$ and
$k = m-1$. A fault-tolerant path connecting the copy of $s$ in $G_1$ to the copy of $s$ in $G_2$ corresponds
to a feasible solution to the $m$-SDST instance with the same cost, and vise-verse.


\section{Polynomial Special Cases}\label{sec:exact}

This section is concerned with tractable restrictions of FTP. Concretely we give polynomial
algorithms for arbitrary graphs and $k=1$, directed acyclic graphs (DAGs) and fixed $k$ and
for Series-parallel graphs. We denote the problem FTP restricted to instances with some 
fixed $k$ by $k$-FTP.

\subsection{$1$-FTP}\label{subsec:kisone}

For simplicity of the exposition we assume that $G$ is undirected. For the adaptation
to directed graphs see Remark~\ref{rem:dadskisone}.
We start with the following useful structural lemma.

\begin{lemma}\label{lem:disjointpaths}
 Let $X^*$ be an optimal solution to FTP on the instance $(G,M,k)$. Let $H' = (V', X')$ be the 
graph obtained by contracting all edges of $X^* \setminus M$ in $H^* = (V, X^*)$. Let $s'$ and
$t'$ be the vertices in $V'$, to which $s$ and $t$ were contracted, respectively.
Then either $s' = t'$, or there are $k+1$ edge disjoint $s'$-$t'$ paths in $H'$.
\end{lemma}

\begin{proof}
If $X^* \subset E\setminus M$ then $s' = t'$. Assume this is not the case, namely that
$X^* \cap M \neq \emptyset$. Observe that $H'$ must contain at least $k+1$ edge
disjoint $s'$-$t'$ paths: indeed if there was a cut of cardinality $k$ or less in $H$,
this would also be a cut in $H^*$. Since $X' \subset M$, it corresponds to a failure scenario,
namely $X' \in \Omega(M,k)$.
\end{proof}

\begin{definition}\label{def:ercc:unionpaths}
An $s$-$t$ \textit{bipath} in the graph $G =(V,E)$ is a union of two $s$-$t$ paths $P_1, P_2 \subset E$, such that
for every two nodes $u,v$ incident to both $P_1$ and $P_2$, the order in which they appear on $P_1$ and $P_2$, respectively,
when traversed from $s$ to $t$ is the same.
\end{definition}


In the context of $1$-FTP, we call a bipath $Q=P_1 \cup P_2$ \textit{robust}, if it holds that $P_1 \cap P_2 \cap M = \emptyset$.
Note that every robust $s$-$t$ bipath $Q$ in $G$ is a feasible solution to the $1$-FTP instance.
Indeed, consider any failure edge $e\in M$. Since $e\not\in P_1\cap P_2$ it holds that either $P_1 \subset Q-e$,
or $P_2 \subset Q-e$. It follows that $Q-e$ contains some $s$-$t$ path.
The next lemma shows that every feasible solution of the $1$-FTP instance contains a robust $s$-$t$ bipath.

\begin{lemma}\label{lem:robustbipath_contained}
Every feasible solution $S^*$ to the $1$-FTP contains a robust $s$-$t$ bipath.
\end{lemma}

\begin{proof}
We assume without loss of generality that $S^*$ is a minimal feasible solution with respect
to inclusion. Let $Y \subset S^*$ be the set of bridges in $(V, S^*)$. From feasibility of 
$S^*$, we have $Y \cap M = \emptyset$. Consider any $s$-$t$ path $P$ in $S^*$. Let $u_1, \cdots, u_r$ be
be the set of vertices incident to $Y = P\cap Y$. Let $u_i$ and $u_{i+1}$ be such that $u_iu_{i+1} \not\in Y$.
(if such an edge does not exist, we have $Y = P$, which means that $P$ is a robust $s$-$t$ bipath).
Note that $S^*$ must contain two edge-disjoint $u_i$-$u_{i+1}$ paths $L_1, L_2$. Taking as
the set $Y$ together with all such pairs of paths $L_1, L_2$ results in a robust bipath.
\end{proof}

We can conclude from the previous discussion and Lemma~\ref{lem:robustbipath_contained} that all
minimal feasible solutions to the $1$-FTP instance are robust bipaths. This observations leads
to the simple algorithm, which is formally given as Algorithm~\ref{alg:kisone}. We prove its
correctness in the following theorem.

\begin{theorem}\label{thm:kisone}
There is a polynomial algorithm for $1$-FTP.
\end{theorem}

\begin{proof}
To solve $1$-FTP we need to find the minimum cost robust $s$-$t$ bipath. To this end 
let us define two length functions $\ell_1, \ell_2 : V\times V \rightarrow \mathbb{R}_+$. For two
vertices $u,v\in V$ let $\ell_1(u,v)$ denote the shortest path distance from $u$ to $v$ in
the graph $(V,E\setminus M)$, and let $\ell_2(u,v)$ denote the cost of the shortest pair of edge-disjoint
$u$-$v$ paths in $G$. 
Clearly, both length functions can be computed in polynomial time (e.g. using
flow techniques). Finally, set $\ell(u,v) = \min\{\ell_1(u,v), \ell_2(u,v)\}$. 
Construct the complete graph on the vertex set $V$ and associate the length function $\ell$ with it. 
Observe that by definition of $\ell$, any $s$-$t$ path in this graph corresponds to a robust
$s$-$t$ bipath with the same cost, and vice versa. It remains to find the shortest $s$-$t$ bipath
by performing a single shortest $s$-$t$ path in the new graph. For every edge $uv$ in this 
shortest path, the optimal bipath contains the shortest $u$-$v$ path in $(V,E\setminus M)$ if
$\ell(u,v) = \ell_1(u,v)$, and the shortest pair of $u$-$v$ paths in $G$, otherwise.
\end{proof}

\renewcommand{\algorithmicrequire}{\textbf{Input:}}
\renewcommand{\algorithmicensure}{\textbf{Output:}}

\begin{algorithm}
\caption{\textbf{: FindBipath($G, s, t, M$)}}\label{alg:kisone}
\begin{algorithmic}[1]
\REQUIRE{$G=(V,E)$ a directed graph, $s,t\in V$ and $M\subset E$.}
\ENSURE{A minimum-cost $s$-$t$ bipath $Q$.}
\STATE{$Q \gets \emptyset.$}
\STATE{$G_{SAFE} \gets (V,E\setminus M).$}
\FORALL{$u,v \in V\times V$}
  \STATE{Compute $\ell_1(u,v)$, the shortest $u$-$v$ path distance in $G_{SAFE}$.}
  \STATE{Compute $\ell_2(u,v)$, the cost of the minimum $u$-$v$ $2$-flow in $G$ with unit capacities.}
  \STATE{$\ell(u,v) \gets \min\{\ell_1(u,v), \ell_2(u,v)\}$.}
\ENDFOR
\STATE{Compute the shortest $s$-$t$ path $P$ in $(V, V\times V)$ and length function $\ell$.}
\FORALL{$uv \in P$}
  \IF{$\ell(u,v) = \ell_1(u,v)$}
    \STATE{Add to $Q$ a shortest $u$-$v$ path in $G_{SAFE}$.}
  \ELSE 
    \STATE{Add to $Q$ a shortest edge-disjoint pair of $u$-$v$ paths in $G$.}
  \ENDIF
\ENDFOR
\RETURN{$Q.$}
\end{algorithmic}
\end{algorithm}

\begin{remark}\label{rem:dadskisone}
Let us remark that Algorithm~\ref{alg:kisone} easily extends to work on directed graphs as follows. 
Lemma~\ref{lem:disjointpaths} is replaced with the more general statement that the minimum
$s$-$t$ flow in the graph $(V, X^*)$ capacitated by the vector $c_e = 1$ if $e\in M$ and $c_e = \infty$,
otherwise is at least $k+1$. The rest of the analysis remains the same. 
\end{remark}

\subsection{$k$-FTP and Directed Acyclic Graphs}\label{subsec:dagfixedk}

Let us first consider the case of a layered graph. The generalization to a
directed acyclic graph is done via a standard transformation, which we describe later.
Recall that a layered graph 
$G = (V, E)$ is a graph with a partitioned vertex set $V = V_1 \cup \cdots\cup V_r$
and a set of edges satisfying $E \subset \bigcup_{i\in [r-1]} V_i \times V_{i+1}$.
We assume without loss of generality that $V_1 = \{s\}$ and $V_{r} = \{t\}$. For every $i\in [r-1]$ 
we let $E_i = E \cap V_i \times V_{i+1}$.

Analogously to the algorithm in the previous section, we reduce $k$-FTP to a shortest
path problem in a larger graph. The following definition sets the stage for the
algorithm.

\begin{definition}\label{def:conf}
 An \emph{$i$-configuration} is a vector $d \in \{0,1,\cdots,k+1\}^{V_i}$ 
satisfying $\sum_{v\in V_i} d_v = k+1$. We let $supp(d) = \{v\in V_i : d_v > 0\}$.
For an $i$-configuration $d^1$ and an $(i+1)$-configuration $d^2$ we let 
$$
V(d^1, d^2) = supp(d^1) \cup supp(d^2) \,\, \text{ and } \,\, E(d^1, d^2) = E[V(d^1, d^2)].
$$

We say that an $i$-configuration $d^1$ \emph{precedes} an $(i+1)$-configuration $d^2$ if the 
following flow problem is feasible. The graph is defined as $H(d^1, d^2) = (V(d^1, d^2), E(d^1, d^2))$.
The demand vector $\nu$ and the capacity vector $c$ are given by 
\begin{equation*}
\nu_u = 
\begin{cases} 
-d^1_u    & \mbox{if }  u\in supp(d^1)\\ 
d^2_u      & \mbox{if }  u\in supp(d^2) 
\end{cases} \,\, \text{and} \,\,\,\,
c_e = 
\begin{cases} 
1    	    & \mbox{if }  e\in M\\ 
\infty      & \mbox{if } e\in E\setminus M,
\end{cases}
\end{equation*}
respectively. If $d^1$ precedes $d^2$ we say that the \emph{link} $(d^1,d^2)$ exists. Finally, the \emph{cost}
$\ell(d^1,d^2)$ of this link is set to be minimum value $w(E')$ over all $E'\subset E(d^1, d^2)$,
for which the previous flow problem is feasible, when restricted to the set of edges $E'$.
\end{definition}

The algorithm constructs a layered graph $\mathcal{H} = (\mathcal{V}, \mathcal{E})$ with $r$ layers 
$\mathcal{V}_1, \cdots, \mathcal{V}_r$. For every $i\in[r]$ the set of 
vertices $\mathcal{V}_i$ contains all $i$-configurations. Observe that $\mathcal{V}_1$ and $\mathcal{V}_r$ 
contain one vertex each, denoted by $c^s$ and $c^t$, respectively. The edges correspond to links between
configurations. Every edge is directed from the configuration with the lower index to the one with the higher index.
The cost is set according to Definition~\ref{def:conf}.
The following lemma provides the required observation, which immediately leads to a polynomial algorithm.

\begin{lemma}
Every $c^s$-$c^t$ path $P$ in $H$ corresponds to a fault-tolerant path $S$ with $w(S) \leq \ell(P)$, and vise-versa.
\end{lemma}

\begin{proof}
Consider first a fault-tolerant path $S \subset E$. We construct a corresponding $c^s$-$c^t$ path in $H$
as follows. Consider any $k+1$ $s$-$t$ flow $f^S$, induced by $S$. Let $p^1, \cdots, p^l$ be a path 
decomposition of $f^S$ and let $1 \leq \rho_1, \cdots, \rho_l \leq k+1$ (with $\sum_{i\in [l]} \rho_i = k+1$) be the
corresponding flow values.

Since $G$ is layered, the path $p^j$ contains exactly one vertex $v^j_i$ from $V_i$ and one edge $e^j_i$ 
from $E_i$ for every $j\in [l]$ and $i\in [r]$. For every $i\in[r]$ define the $i$-configuration $d^i$ with
$$
d^i_v = \sum_{j\in[l] : v = v^i_j} \rho_i,
$$
if some path $p^j$ contains $v$, and $d^i_v = 0$, otherwise. The fact that $d^i$ is an $i$-configuration follows
immediately from the fact that $f^S$ is a $(k+1)$-flow. In addition, for the same reason 
$d^i$ precedes $d_{i+1}$ for every $i\in [r-1]$. From the latter observations and the fact that $d^1 = c^s$
and $d^r =  c^t$ it follows that $P = d^1, d^2, \cdots, d^r$ is a $c^s-c^t$ path in $H$ with cost $\ell(P) \leq w(S)$.

Consider next an $c^s-c^t$ path $P = d^1, \cdots, d^r$ with cost $\ell(P) = \sum_{i=1}^{r-1} \ell(d^i, d^{i+1})$.
The cost $\ell(d^i, d^{i+1})$ is realized by some set of edges $R_i \subset E(d^i, d^{i+1})$ for every $i\in [r-1]$.
From Definition~\ref{def:conf}, the maximal $s$-$t$ flow in the graph $G' = (V, R)$ is at least $k+1$, where 
$R = \cup_{i\in [r-1]} R_i$. Next, Lemma~\ref{lem:disjointpaths} guarantees that there exists some feasible solution
$S \subset R$, the cost of which is at most $\ell(P)$. In the latter claim we used the disjointness of the sets 
$R_i$, which is due the layered structure of the graph $G$. This concludes the proof of the lemma.
\end{proof}

Finally, we observe that the number of configurations is bounded by $O(n^{k+1})$, which implies that
$k$-FTP can be solved in polynomial time on layered graphs. 

To obtain the same result for directed
acyclic graphs we perform the following transformation of the graph. Let $v_1, \cdots, v_n$
be a topological sorting of the vertices in $G$. Replace every edge $e = v_iv_j$ ($i<j$) with a path 
$p_e = v_i, u^e_{i+1}, \cdots, u^e_{j-1}, v_j$ of length $j-i+1$ by subdividing it sufficiently many times. Set
the cost of the first edge on the path to $w'(v_iu^e_{i+1}) = w(v_iv_j)$ and set the costs of all other edges on the path
to zero. In addition, create a new set of faulty edges $M'$, which contains all edges in a path $p_e$ if $e\in M$.

It is straightforward to see that the new instance of FTP is equivalent to the original one, while the obtained
graph after the transformation is layered. We summarize with the following theorem.

\begin{theorem}\label{thm:dagpoly}
There is a polynomial algorithm for $k$-FTP restricted to instances with a directed acyclic graph.
\end{theorem}

\subsection{Series-Parallel Graphs}\label{subsec:srp}

Recall that a graph is called \emph{series-parallel (SRP)} with terminal $s$ and $t$ if 
it can be composed from a collection of disjoint edges using the \emph{series} and \emph{parallel} compositions.
The series composition of two SRP graphs with terminals $s$, $t$ and $s',t'$ respectively, takes
the disjoint union of the two graphs, and identifies $t$ with $s'$. The parallel composition 
takes the disjoint union of the two graphs and identifies $s$ with $s'$ and $t$ with $t'$.
Given a SRP graph it is easy to obtain the aforementioned decomposition.

The algorithm we present has linear running time whenever the robustness parameter $k$ is fixed.
The algorithm is given as Algorithm~\ref{alg:ercc:sircc:srp}. 
In fact, the algorithms computes the optimal solutions $S_{k'}$ for all parameters
$0 \leq k' \leq k$. The symbol $\perp$ is returned if the problem is infeasible.

\begin{algorithm}
\caption{\textbf{: FTP-SeriesParallel($G,s,t,k$)}}\label{alg:ercc:sircc:srp}
\begin{algorithmic}[1]
\REQUIRE{$G=(V,E)$ a series-parallel graph, $s,t\in V$ and $M\subset E$, $k\in \mathbb{Z}_+$.}
\ENSURE{Optimal solution to FTP for parameters $0,1,\cdots, k$.}
\IF{$E = \{e\} \wedge e\in M$}
  \STATE Return $(\{e\}, \perp, \cdots, \perp)$
\ENDIF
\IF{$E = \{e\} \wedge e\not\in M$}
  \STATE Return $(\{e\}, \cdots, \{e\})$
\ENDIF

\vspace{2mm}

$\Rightarrow$ $G$ is a composition of $H_1, H_2$.
\vspace{2mm}

\STATE $(S_0^1, \cdots, S_k^1) \leftarrow $ \textbf{FTP-SeriesParallel}$(H_1, M\cap E[H_1], k)$
\STATE $(S_0^2, \cdots, S_k^2) \leftarrow $ \textbf{FTP-SeriesParallel}$(H_2, M\cap E[H_2], k)$

\IF{$G$ \text{is a series composition of} $H_1, H_2$}

\FOR{$i=0,\cdots,k$}
  \IF{$S_i^1 = \perp \vee \, S_i^2 = \perp$}
    \STATE $S_i \leftarrow \perp$
  \ELSE 
    \STATE $S_i \leftarrow S_i^1 \cup S_i^2$
  \ENDIF
\ENDFOR

\ENDIF

\IF{$G$ \text{is a parallel composition of} $H_1, H_2$}
  \STATE $m_1 \leftarrow \max\{i : S_i^1 \neq \perp\}$
  \STATE $m_2 \leftarrow \max\{i : S_i^2 \neq \perp\}$

  \FOR{$i=0,\cdots,k$}
    \IF{$i > m_1 + m_2 + 1$} 
      \STATE $S_i \leftarrow \perp$
    \ELSE 
      \STATE $r \leftarrow argmin_{-1\leq j\leq i}\{w(S_j^1) + w(S_{i-j-1}^2)\}$ \hspace{20mm} $// \, S^1_{-1} = S^2_{-1} := \emptyset$
      \STATE $S_i \leftarrow S_r^1 \cup S_{i-r-1}^2$
    \ENDIF
  \ENDFOR
\ENDIF

\STATE Return $(S_0, \cdots, S_k)$
\end{algorithmic}
\end{algorithm}

\begin{theorem}\label{thm:fastalgSRP}
Algorithm~\ref{alg:ercc:sircc:srp} returns an optimal solution to the FTP problem on SRP
graphs. The running time of Algorithm~\ref{alg:ercc:sircc:srp} is $O(nk)$.
\end{theorem}

\begin{proof}
 The proof of correctness is by induction on the depth of the recursion in Algorithm~\ref{alg:ercc:sircc:srp}. Clearly 
the result returned by Algorithm~\ref{alg:ercc:sircc:srp} in lines $1$-$6$ is optimal.
Assume next that the algorithm computed correctly all optimal solutions for the subgraphs $H_1, H_2$,
namely that for every $i\in [2]$ and $j\in [k]$, the set $S_j^i$ computed in lines $7$-$8$ is an
optimal solution to the problem on instance $\mathcal{I}_i^j = (H_i, M \cap E[H_i], j)$.

Assume first that $G$ is a series composition of $H_1$ and $H_2$, and let $0 \leq i \leq k$. If
either $S_i^1 = \perp$ or $S_i^2 = \perp$ the problem with parameter $i$ is clearly also infeasible,
hence the algorithm works correctly in this case. Furthermore, since $G$ contains a cut vertex (the
terminal node, which is in common to $H_1$ and $H_2$), a solution $S$ to the problem is feasible for
$G$ if and only if it is a union of two feasible solutions for $H_1$ and $H_2$. From the inductive
hypothesis it follows that $S_i$ is computed correctly in line $14$.

Assume next that $G$ is a series composition of $H_1$ and $H_2$. Consider any feasible solution $S'$
to the problem on $G$ with parameter $i$. Let $S'_1$ and $S'_2$ be the restrictions of $S'$ to edges
of $H_1$ and $H_2$ respectively, and let $n_1$ and $n_2$ be the maximal integers such that
$S'_1$ and $S'_2$ are robust paths for $H_1$ and $H_2$ with parameters $n_1$ and $n_2$, respectively.
Observe that $i \leq n_1 + n_2 + 1$ must hold. Indeed if this would not be the case, then taking
any cut with $n_1 + 1$ edges in $S'_1$ and another cut with $n_2 + 1$ edges in $S'_2$ yields a cut
with $n_1 + n_2 + 2$ edges in $G$, contradicting the fact that $S'$ is a robust path with parameter $i$.
We conclude that the algorithm computes $S_i$ correctly in line $23$. Finally note
that the union any two robust paths for the graphs $H_1$ and $H_2$ with parameters $n_1$ and $n_2$ with 
$i \leq n_1 + n_2 + 1$ yield a feasible solution $S_i$. It follows that the minimum cost such robust
path is obtained as a minimum cost of a union of two solutions for $H_1$ and $H_2$, with robustness
parameters $j$ and $i-j-1$ for some value of $j$. To allow $S_i = S_i^1$ or $S_i = S_i^2$ we let 
$j$ range from $-1$ to $k$ and set $S^1_{-1} = S^2_{-1} = \emptyset$. This completes the proof of correctness.

To prove the bound on the running time, let $T(m,k)$ denote the running time of the algorithm on
a graph with $m$ edges and robustness parameter $k$. We assume that the graph is given by a 
hierarchical description, according to its decomposition into single edges.
The base case obviously takes $O(k)$ time. Furthermore we assume that the solution $(S_0, \cdots, S_k)$
is stored in a data structure for sets, which uses $O(1)$ time for generating empty sets and for
performing union operations.
If the graph is a series composition then the running time satisfies 
$T(m,k) \leq T(m',k) + T(m-m',k) + O(k)$ for some $m' < m$. If the graph is a parallel composition, then $T(m,k)$
satisfied the same inequality. We assume that the data structure, which stores the sets $S_i$
also contains the cost of the edges in the set. This value can be easily updates in time $O(1)$
when the assignment into $S_i$ is performed. It follows that $T(m,k) = O(mk) = O(nk)$ as required.
\end{proof}

\section{Integrality Gap and Approximation Algorithms}\label{sec:intgap}

In this section we study the natural fractional relaxation of FTP. We prove a tight
bound on the integrality gap of this relaxation. This results also suggests a simple
approximation algorithm for FTP with ratio $k+1$. We later combine this algorithm 
with Algorithm~\ref{alg:kisone} to obtain a $k$-approximation algorithm.

\subsection{Fractional FTP and Integrality Gap}

Let us start by introducing the fractional relaxation of FTP, which we denote by FRAC-FTP. The
input to FRAC-FTP is identical to the input to FTP. The goal in FRAC-FTP is to find a capacity
vector $x:E \rightarrow [0,1]$ of minimum cost $w(x) = \sum_{e\in E} w_e x_e$ such that for
every $F\in \Omega(M,k)$, the maximum $s$-$t$ flow in $G-F$, capacitated by $x$ is at least one. 
Note that by the Max-Flow Min-Cut Theorem, the latter condition is equivalent to requiring that 
the minimum $s$-$t$ cut in $G-F$ has capacity of at least one. We will use this fact in the proof
of the main theorem in this section.

Observe that by requiring $x\in \{0,1\}^E$ we obtain FTP, hence FRAC-FTP is indeed a fractional relaxation
of FTP.

In the following theorem by 'integrality gap' we mean the maximum ratio between the optimal solution
value to an FTP instance, and the optimal value of the corresponding FRAC-FTP instance.

\begin{theorem}\label{thm:integralitygap}
The integrality gap of FTP is bounded by $k+1$. Furthermore, there exists an infinite family of
instances of FTP with integrality gap arbitrarily close to $k+1$.
\end{theorem}

\begin{proof}
Consider an instance $I = (G,M,k)$ of FTP. Let $x^*$ denote an optimal solution to the corresponding
FRAC-FTP instance, and let $OPT = w(x^*)$ be its cost. Define a vector $y\in \mathbb{R}^E$ as follows.
\begin{equation}\label{eq:gapproof1}
y_e = 
\begin{cases} 
(k+1) x_e    & \mbox{if } e\not\in M \\ 
\min\{1, (k+1)x_e\}  & \mbox{otherwise}.
\end{cases}
\end{equation}
Clearly, it holds that $w(y) \leq (k+1)OPT$. 
We claim that every $s$-$t$ cut in $G$ with capacities $y$ has capacity of at least $k+1$.
Consider any such cut $C \subset E$, represented as the set of edges in the cut. Let 
$M' = \{e\in M : x^*_e \geq \frac{1}{k+1}\}$ denote the set of faulty edges attaining high fractional
values in $x^*$. Define $C' = C\cap M'$. If $|C'| \geq k+1$ we are clearly done. Otherwise, assume $|C'| \leq k$.
In this case consider the failure scenario $F = C'$. Since $x^*$ is a feasible solution it must hold that
\begin{equation}\label{eq:gapproof2}
\sum_{e\in C\setminus C'} x^*_e \geq 1.
\end{equation}
Since for every edge $e\in C \setminus C'$ it holds that $y_e = (k+1)x^*_e$ we obtain
\begin{equation}\label{eq:gapproof3}
\sum_{e\in C\setminus C'} y_e \geq k+1,
\end{equation} 
as desired. From our observations it follows that the maximum flow in $G$ with capacities $y$ is at least
$k+1$. Finally, consider the minimum cost $(k+1)$-flow $z^*$ in $G$ with capacities defined by
\begin{equation}\label{eq:gapproof4}
c_e = 
\begin{cases} 
k+1    & \mbox{if } e\not\in M \\ 
1  & \mbox{otherwise}.
\end{cases}
\end{equation}
From integrality of $c$ and the minimum-cost flow problem we can
assume that $z^*$ is integral. Note that $y_e \leq c_e$ for every $e\in E$, hence any feasible $(k+1)$-flow with capacities
$y$ is also a feasible $(k+1)$-flow with capacities $c$. From the previous observation it holds that $w(z^*) \leq w(y) \leq (k+1)OPT$.
From Lemma~\ref{lem:disjointpaths} we know that $z^*$ is a feasible solution to the FTP instance.
This concludes the proof of the upper bound of $k+1$ for the integrality gap. 

To prove the same lower bound we provide an infinite family of instances, containing instances with integrality
gap arbitrarily close to $k+1$. Consider a graph with $D \gg k$ parallel edges with unit cost connecting $s$ and $t$, and let $M=E$. 
The optimal solution to FTP
on this instance chooses any subset of $k+1$ edges. At the same time, the optimal solution to FRAC-FTP assigns
a capacity of $\frac{1}{D-k}$ to every edge. This solution is feasible, since in every failure scenario, the number
of edges that survive is at least $D-k$, hence the maximum $s$-$t$ flow is at least one. The cost of this solution is $\frac{D}{D-k}$.
Taking $D$ to infinity yields instances with integrality gap arbitrarily close to $k+1$.
\end{proof}

The proof of Theorem~\ref{thm:integralitygap} implies a simple $(k+1)$-approximation algorithm for
FTP. This algorithm simply solves the integer minimum-cost flow problem, defined in proof of the theorem, and returns 
the set of edges corresponding to the optimal integral flow $z^*$ as the solution. 
This result is summarized in the following corollary.

\begin{corollary}\label{cor:kplusone_pprox}
There is a polynomial $(k+1)$-approximation algorithm for FTP.
\end{corollary}

\subsection{A $k$-Approximation Algorithm}\label{subsec:kapprox}

In this section we improve the approximation algorithm from the previous section.
The new algorithm can be seen as a generalization of Algorithm~\ref{alg:kisone} to arbitrary FTP instances. The main
observation is the following. The reason why the approximation algorithm implied by Theorem~\ref{thm:integralitygap}
gives an approximation ratio of $k+1$ is that the capacity of edges in $E\setminus M$ is set to $k+1$, hence,
if the flow $z^*$ uses such edges to their full capacity, the cost incurred is $k+1$ times the cost of these
edges. This implies that the best possible lower bound on the cost $w(z^*)$ is $(k+1)OPT_{FRAC}$, where 
$OPT_{FRAC}$ denotes the optimal solution value of the corresponding FRAC-FTP instance. 
To improve the 
algorithm we observe that the edges in $z^*$, which carry a flow of $k+1$ are cut-edges in the obtained solution.

To conveniently analyze our new algorithm let us consider a certain canonical flow defined 
by minimal feasible solutions.

\begin{definition}\label{def:canonicalflow}
Consider an inclusion-wise minimal feasible solution $S\subset E$ of an instance $I = (G,s,t,M)$ of FTP. 
A \emph{flow $f^S$ induced by $S$} is any integral $s$-$t$ $(k+1)$-flow in $G$ respecting the capacity vector
$$
c^S_e = 
\begin{cases} 
1    & \mbox{if }  e \in S\cap M\\ 
k+1  & \mbox{if }  e \in S \setminus M\\ 
0    & \mbox{if }   e\in E\setminus S.
\end{cases}
$$
\end{definition}

%
%
%

To this end consider an optimal solution $X^* \subset E$ to the FTP instance and consider any corresponding 
induced flow $f^*$. Define 
$$X_{PAR} = \{e\in X^* : f^*(e) \leq k\} \,\,\, \text{and} \,\,\,  X_{BRIDGE} = \{e\in X^* : f^*(e) = k+1\}.$$
As we argued before, every edge in $X_{BRIDGE}$ must be a bridge in $H = (V,X^*)$ disconnecting $s$ and $t$.
Let $u_e$ denote the tail vertex of an edge $e\in E$. Since every edge $e\in X_{BRIDGE}$ constitutes
an $s$-$t$ cut in $H$, it follows that the vertices in $U = \{e_u : e\in X_{BRIDGE}\} \cup \{s,t\}$ can be unambigously 
ordered according to the order in which they appear on any $s$-$t$ path in $H$, traversed from $s$ to $t$. 
Let $s = u_1, \cdots, u_q = t$ be this order. Except for $s$ and $t$, every vertex in $U$ constitutes
a cut-vertex in $H$. Divide $H$ into $q-1$ subgraphs $H^1, \cdots, H^{q-1}$ by letting $H^i = (V, Y_i)$ contain the 
union of all $u_i$-$u_{i+1}$ paths in $H$. We observe the following property.

\begin{proposition}\label{prop:decomp}
For every $i\in [q-1]$ the set $Y_i\subset E$ is an optimal solution to the FTP instance $I_i = (G,u_i, u_{i+1}, M)$.
\end{proposition}

Consider some $i\in[q-1]$ and let $f^*_i$ denote the flow $f^*$, restricted to edges in $H^i$. Note that
$f^*_i$ can be viewed as a $u_i$-$u_{i+1}$ $(k+1)$-flow.
Exactly one of the following cases can occur. Either $H^i$ contains a single edge $e\in E\setminus M$, or 
$$
\max_{e\in Y_i}{f^*_i(e)} \leq k.
$$
In the former case, the edge $e$ is the shortest $u_i$-$u_{i+1}$ path in $(V,E\setminus M)$. In the latter
case we can use a slightly updated variant of the algorithm in Corollary~\ref{cor:kplusone_pprox} to obtain
a $k$-approximation of the optimal FTP solution on instance $I_i$. Concretely, the algorithm defines 
the capacity vector
$$
c'_e = 
\begin{cases} 
k    & \mbox{if } e\not\in M \\ 
1  & \mbox{otherwise},
\end{cases}
$$
and finds an integral minimum-cost $u_i$-$u_{i+1}$ $(k+1)$-flow $Y^*$ in $G$, and returns the support $Y\subset E$
of the flow as the solution. The existence of the flow $f^*_i$ guarantees that $w(y^*) \leq w(f^*_i)$,
while the fact that the maximum capacity in the flow problem is bounded by $k$ gives $w(Y) \leq kw(y^*)$.
It follows that this algorithm approximates the optimal solution to the FTP instance $I_i$ to within a factor $k$.

To describe the final algorithm it remains use the blueprint of Algorithm~\ref{alg:kisone}. The only difference
is in the computation performed in step $5$. Instead of finding two edge-disjoint $u$-$v$ paths, the new
algorithm solves the aforementioned flow problem. An appropriate adaptation is also made to step $12$.
We summarize the main result of this section in the following theorem. The proof is omitted, as it is identical
to that of Theorem~\ref{thm:kisone}, with the exception of the preceding discussion.

\begin{theorem}\label{thm:kapprox}
There is a polynomial $k$-approximation algorithm for FTP.
\end{theorem}

\section{Conclusions and Future Work}\label{sec:conclusions}

This paper presents FTP, a simple natural extension of the minimum $k$ edge-disjoint path problem,
interpreted as a fault-tolerant connection problem. In our model not all subsets of $k$ edges 
can be removed from the graph after a solution is chosen, but rather any subset of \emph{faulty edges},
a subset provided as part of the input. Such an adaptation is natural from the point of view of
many application domains.

We observed a dramatic increase in the computational complexity of FTP with respect EDP.
At the same time we identified several classes of instances admitting a polynomial exact algorithm.
These classes include the case $k=1$, directed acyclic graphs and fixed $k$ and series-parallel graphs.
Next, we defined a fractional counterpart of FTP and proved a tight bound on the corresponding integrality
gap. This result lead to a $k$-approximation algorithm for FTP. We conclude by stating a number
of promising directions for future research.

Perhaps the main problem remains to understanding the approximability of FTP. In particular,
it is interesting to see if the approximation guarantee for FTP can be decreased to the
approximation guarantees of the best known algorithms for the Steiner Tree problem. 
It is also interesting to relate FTP to more general problems such a Minimum-Cost Fixed-Charge 
Network Flow and special cases thereof. 
The complexity of $k$-FTP in still unknown. It is interesting to see if the methods
employed in the current paper for $1$-FTP and $k$-FTP on directed acyclic graphs can be 
extended to $k$-FTP on general graphs.
It is also natural to study other classical combinatorial optimization problems in the
model proposed in this paper. In particular, the counterpart of the Minimum Spanning Tree
problem seems to be a promising next step. We remark that all the results about integrality
gaps and approximability of FTP can be carried over (with some modifications) to this problem.
For example, the proof of Theorem~\ref{thm:integralitygap} can be adapted to the latter
problem to give an integrality gap of $k+3$, in combination with the integrality gap result
for the Minimum-Cost $k$-Edge Connected Spanning Subgraph problem of Gabow et al.~\cite{kEdgeConnected1}.

\bibliographystyle{plain}



\end{document}